\documentclass{article}[12pt]
\RequirePackage{amsthm,amsmath,amsfonts,amssymb}
\RequirePackage[authoryear]{natbib}
\RequirePackage[colorlinks,citecolor=blue,urlcolor=blue]{hyperref}
\usepackage[margin=1in]{geometry}
\usepackage{setspace}
\theoremstyle{plain}
\newtheorem{assumption}{Assumption}
\newtheorem{theorem}{Theorem}

\newtheorem{lemma}{Lemma}

\usepackage{graphicx,psfrag,epsf}
\usepackage{url}
\usepackage{amsthm,amsmath,amsfonts,amssymb}
\usepackage{bm}
\usepackage{array}
\usepackage{diagbox}
\usepackage{lscape}
\usepackage{multirow}
\usepackage{threeparttable}
\usepackage{enumitem}
\usepackage{tikz}
\usepackage{wrapfig}
\usepackage{multirow}

\usepackage{venndiagram}
\newcommand{\probP}{\text{I\kern-.23em P}}

\usepackage{makecell}
\newcolumntype{P}[1]{>{\centering\arraybackslash}p{#1}}
\usepackage{afterpage}

\begin{document}
\begin{center} {\bf\Large Improving the Power of Bonferroni Adjustments \\under Joint Normality and Exchangeability }
\end{center}
\centerline{\textsc{Caleb Hiltunen and Yeonwoo Rho\footnote{emails: C. Hiltunen (cahiltun@mtu.edu), Y. Rho (yrho@mtu.edu)}}}		
		\bigskip
	\centerline {Michigan Technological University}
	\bigskip
\begin{abstract}
Bonferroni's correction is a popular tool to address multiplicity but is notorious for its low power when tests are dependent. This paper proposes a practical modification of Bonferroni's correction when test statistics are jointly normal and exchangeable. This method is intuitive to practitioners and achieves higher power in sparse alternatives, as our simulations suggest. We also prove that this method successfully controls the family-wise error rate at any significance level. \\
		\textit{key words}:  Bonferroni's correction, Multiple comparisons, Exchangeable, Jointly normal test statistics.
\end{abstract}
\section{Introduction}\label{section:intro}
Multiple comparison correction procedures are becoming increasingly important with the prevalence of large datasets. When it comes to controlling family-wise error rate (FWER), Bonferroni's correction \citep{bonf36}, for its ease of use, remains one of the popular methods among practitioners \citep{polanin14}, despite the existence of alternative approaches such as \cite{tippett31,fisher32,pearson33,stouffer49,wilkinson51}. However, Bonferroni is also well-known for its conservatism, particularly when the number of test statistics becomes very large or when tests are dependent \citep{chen17}. Currently,  other methods exist which are not as conservative as Bonferroni in the face of dependence, such as \cite{holm79}, \cite{benjamini95}, \cite{proshawn11}, \cite{wilson19}, and \cite{liu20}.

In this paper, we modify Bonferroni's correction to improve its power under dependence. 
We limit our scope to jointly normal test statistics that are invariant to permutations, i.e., exchangeable. We also assume test statistics are standardized with variance 1. These assumptions are not too uncommon in the $p$-value combination literature  \citep{gupta73, tian21, choi23, gasparin2025}, as many null distributions of test statistics are standard normal, and exchangeability is the simplest assumption for dependence when there are no specific orders among the test statistics. 
By limiting its scope, our method keeps the intuitive nature of Bonferroni's and achieves a significant improvement in power, particularly under sparse alternatives, without assuming the knowledge of any unknown parameters.

Our method is anchored on \cite{gupta73}, where a modification of Bonferroni's is proposed under joint normality and exchangeability. However, their method requires information on the common pairwise correlation of $\rho$, which, in practice, is unknown. Estimation of $\rho$ has been somewhat dismissed in the $p$-value combination literature, particularly when there is only one set of test statistics, perhaps because the absence of replication typically prevents consistent estimation via the weak law of large numbers. 

We propose to estimate $\rho$ using the sample variance. This is possible because one set of positively correlated test statistics under the null would appear to be clustered around a single random number. The magnitude of the correlation determines the tightness of the cluster, which can be quantified by the sample variance.
In fact, this idea is not entirely new. \cite{follman13} proposed this estimator of $\rho$ in a location test with joint normal data with exchangeability. They developed a likelihood ratio test in this setting, demonstrating a use as O'Brien's test \citep{o1984procedures}. However, \cite{follman13}'s approach has received little attention in the p-value combination literature, and the idea of estimating $\rho$ via sample variance has not been applied to \cite{gupta73}'s method.

This paper showcases the use of  $\rho$ estimation with \cite{gupta73}'s modification of Bonferroni's correction. This method boasts higher power under sparse alternatives, compared to the O'Brien's test in \cite{follman13}. It can also be a practical alternative to recently developed $p$-value combination methods based on heavy-tails \citep{wilson19,liu20} when larger significance level is desired. As a Bonferroni-type procedure, our method allows easier identification of individual significant tests following rejection of the global null, compared to the sum-based tests such as O'Brien's and \cite{wilson19,liu20}.

This paper is organized as follows. Section \ref{alternatives} introduces our method along with the theorems that prove its validity, and Section \ref{section:sim} explores the performance in finite samples. Section \ref{section:conclusion} provides concluding remarks. Proofs are relegated to Appendix \ref{app}. Appendix  \ref{app:table} presents size tables.

\section{Method}\label{alternatives}
Let $X = \left(X_1,...,X_n\right)$ be a sequence of jointly normal, exchangeable, and standardized test statistics with mean $\bm{\mu} \in \mathbb{R}^n$ and positive semi-definite covariance matrix $\Sigma_{\rho} \in \mathbb{R}^n \times \mathbb{R}^n$, where $\rho \in [0,1)$ indicates the common correlation. Consider the following hypothesis test,
$$H_0:\bm{\mu} = \bm{0} \hspace{3cm} H_a:\bm{\mu} \neq \bm{0},$$
where $\bm{0}$ is the zero vector of size $n$. While \cite{gupta73} develop their method for one-sided alternatives, we extend this framework to accommodate two-sided tests. For brevity, we focus our description and simulations on the two-sided case, as the procedures for one-sided tests are similar. We provide Theorem \ref{thm2} detailing how our method can be adapted for one-sided alternatives.

When normal variables are exchangeable, they can be expressed as
\begin{equation*}
    X_i = \sqrt{\rho}Z_0+\sqrt{1-\rho}Z_i \hspace{2cm} i=1,2,...,n,
\end{equation*}
\noindent
where $Z_i$ are $i.i.d$ variates pulled from standard normal distributions for $i=0,1,...,n$. With this representation, we will tweak the method described in \cite{gupta73} for a two-sided alternative using similar methods found in \cite{dunnet55}. We combine our test statistics into a global test statistic,
$M_n\left( \left |X \right| \right) = \max_{1\leq i \leq n}(|X_i|)$, like in the Bonferroni case. We aim for an exact\footnote{Exact control meaning that FWER = $\alpha$ as opposed to strict control meaning FWER $\leq \alpha$.} control of the FWER at some level $\alpha$, so we would then find the critical value $c_{\alpha}$ such that 
\begin{equation*}\label{eq1}
    \begin{aligned}
        P(M_n(|X|)\leq c_\alpha) &= P \left (\max_{i}\left(\left|\sqrt{\rho}Z+\sqrt{1-\rho}Z_i\right| \right) \leq c_\alpha\right) \\
        &=E_Z \left[ P\left( \max_i \left(\left| \sqrt{\rho}Z+\sqrt{1-\rho}Z_i \right|\right) \leq c_\alpha |Z=z\right)\right]\\
        &= \int_{-\infty}^{\infty}\Psi_{\mu = \sqrt{\rho}z, \sigma = \sqrt{1-\rho}}^n\left(c_{\alpha}\right) \phi(z)dz=1-\alpha,
    \end{aligned}
\end{equation*}
where $\Psi_{\mu, \sigma}(\cdot)$ indicates the folded-normal distribution with an underlying normal variable of mean $\mu$ and standard deviation $\sigma$, and $\phi(\cdot)$ indicates the standard normal density function. 
Here, we can then obtain the critical value $c_{\alpha}$ through numerical integration. The global null $H_0$ is rejected when $M_n(|X|) \geq c_{\alpha}$.  Alternatively, the $p$-value, $p_0$, for this global test statistic, $M_n(|X|)$, can be calculated: the global null is rejected when $p_0 \leq \alpha$, where
\begin{equation}\label{gnpa}\tag{GNP*}
p_0(\rho)=1-\int_{-\infty}^{\infty}\Psi_{\mu = \sqrt{\rho}z, \sigma = \sqrt{1-\rho}}^n\left(M_n(|X|)\right) \phi(z)dz.
\end{equation}

\cite{gupta73} provides a table of critical values $c_\alpha$ for the one-sided alternative, assuming $\rho$ is known. However, in practice, $\rho$ is often unknown, making it impossible to use \cite{gupta73}'s table. In this paper, we propose to estimate $\rho$ using \cite{follman13}'s method-of-moments style estimator\footnote{It should be noted that \cite{follman13} define their estimator as $\widehat{\rho} = \max(0, 1-s_n^2)$, which is equivalent to \ref{MOM}. As our framework was developed prior to our awareness of \cite{follman13}, we keep \ref{MOM} in this form to ensure consistency with our theorems and proofs.}:
\begin{equation*}\label{MOM}\tag{MOM}
    \widehat{\rho}_{MOM}=(1-s_n^2)I(s_n^2<1).
\end{equation*}
Here, $s_n^2$ indicates the sample variance of $X$,
$$s_n^2=\frac{1}{n-1}X'\left(I_n-\frac{1}{n}J_n\right)X,$$
where $I_n$ is identity matrix of size $n$ and $J_n$ is the matrix of ones of size $n$. While \cite{follman13} also provide a maximum-likelihood estimator (MLE) of $\rho$, this paper will only consider \ref{MOM}, since the two estimators are asymptotically equivalent, as \cite{follman13} point out. While \ref{gnpa}-\ref{MOM} require $n \rightarrow \infty$, our simulation suggest that $n=20$ is a good benchmark for arbitrary $\alpha$ and $\rho$ such that $\alpha$ is controlled. See Table \ref{table:sizetable} in Appendix \ref{app:table}.

We shall prove that \ref{gnpa}-\ref{MOM} controls FWER. In fact,  any other estimator of $\rho$ that satisfies Assumption \ref{assumption:rho} will work with \ref{gnpa}.
\begin{assumption}\label{assumption:rho}
$\sqrt{2\ln(n)} \left( \sqrt{1-\rho} - \sqrt{1-\widehat{\rho}} \right) \xrightarrow[n\rightarrow \infty]{P}0.$ 
\end{assumption}

\noindent
Theorem \ref{thm1} proves that $\widehat{\rho}_{MOM}$ satisfies Assumption \ref{assumption:rho}.
\begin{theorem}\label{thm1}
Consider a sequence of jointly standard normal random variables $X=(X_1,...,X_n)'$ where
$$E[X_i]=0,\hspace{2cm}E[X_i^2]=1,\hspace{2cm} Corr(X_i,X_j)=\rho$$
for $1\leq i \leq n, 1\leq j \leq n, i\neq j, $ and $0<\rho <1$.
Then, with $\widehat{\rho}$ defined as in \ref{MOM}, we have
$$ \sqrt{a_n} \left( \sqrt{1-\rho} - \sqrt{1-\widehat{\rho}} \right) \xrightarrow[n\rightarrow \infty]{P}0,$$
where $a_n=o(n)$.
\end{theorem}
Since $\ln(n) = o(n)$, Theorem \ref{thm1} implies that \ref{MOM} satisfies Assumption \ref{assumption:rho}. Theorem \ref{thm2} proves that \ref{gnpa}-\ref{MOM} controls the FWER.

\begin{theorem}\label{thm2}
Consider a sequence of jointly standard normal random variables $X=(X_1,...,X_n)'$ where
$$E[X_i]=0,\hspace{2cm}E[X_i^2]=1,\hspace{2cm} Corr(X_i,X_j)=\rho$$
for $1\leq i \leq n, 1\leq j \leq n, i\neq j, $ and $0<\rho <1$.
Then, with $M_n(X)=\max(X_1,...,X_n)$, $|X| = (|X_1|,...,|X_n|)'$, and and an estimator $\widehat{\rho}$ that satisfies Assumption \ref{assumption:rho}, 
\begin{equation*}\label{thm3eq1}
    \left|\int_{-\infty}^{\infty}\left(\Psi_{\mu = \sqrt{\rho}z,\sigma = \sqrt{1-\rho}}^n\left(M_n\left( \left|X \right|\right) \right) -\Psi_{\mu = \sqrt{\widehat{\rho}}z, \sigma = \sqrt{1-\widehat{\rho}}}^n \left( M_n\left( \left|X \right| \right) \right)  \right)\phi(z)dz\right|\xrightarrow[n \rightarrow \infty]{P}0,
\end{equation*}
where $\Psi_{\mu, \sigma}(\cdot)$ is the cumulative distribution function of the folded-normal distribution with location $\mu$ and scale $\sigma$ and $\phi(\cdot)$ is the standard normal density function. Consequently, 
$$p_0(\rho)-p_0(\widehat{\rho}) \xrightarrow[n\rightarrow \infty]{P}0,$$
where $p_0(\cdot)$ is defined in \ref{gnpa}.
\end{theorem}
Because the observed $p$-value, $p_{0}(\widehat{\rho})$, converges to the theoretical $p$-value, $p_{0}(\rho)$, FWER is controlled. For completeness, Theorem \ref{thm3} proves FWER is controlled in the one-sided test case. Without loss of generality, we assume a positive alternative, $H_a:\bm{\mu}>0$ in this Theorem. For negative alternative, $H_a:\bm{\mu}<0$, simply replace $M_n(X)$ with $M_n(-X)$. 
\begin{theorem}\label{thm3}
 Consider a sequence of jointly standard normal random variables $X=(X_1,...,X_n)'$ where
$$E[X_i]=0,\hspace{2cm}E[X_i^2]=1,\hspace{2cm} Corr(X_i,X_j)=\rho$$
for $1\leq i \leq n, 1\leq j \leq n, i\neq j, $ and $0<\rho <1$.
Then, with $M_n=\max(X_1,...,X_n)$ and an estimator $\widehat{\rho}$ that satisfies Assumption \ref{assumption:rho}, 
\begin{equation*}\label{thm2eq2}
    \left| \int_{-\infty}^{\infty}\left(\Phi^n\left( \frac{x\sqrt{\rho}+M_n(X)}{\sqrt{1-\rho}} \right)-\Phi^n\left( \frac{x\sqrt{\widehat{\rho}}+M_n(X)}{\sqrt{1-\widehat{\rho}}} \right) \right)\phi(x)dx\right|\xrightarrow[n \rightarrow \infty]{P}0,
\end{equation*}
implying convergence of the $p$-value for $M_n(X)$.
\end{theorem}

Once the global null is rejected, this Bonferroni-type procedure allows an identification of significant tests. A test statistic $X_i$ will be significant if and only if $X_i \geq c_{\alpha}$, where $c_{\alpha}$ can be found through numerical integration. This method will always yield the maximum, $M_n(|X|)$, as a signficant test statistic, though may be too conservative for the other test statistics. See Figure \ref{Figure 3} for simulation evidence of the conservative nature and Section \ref{section:conclusion} for discussion on potential adjustments.

\section{Simulation Studies}\label{section:sim}
When simulating the exchangeable sequence $X$, we use the following computationally efficient scheme,
$$X_i = \sqrt{\rho}Z_0+\sqrt{1-\rho}Z_i +\mu_i,$$
where $Z_i$ are independent and identically distributed standard normal variates for $0\leq i \leq n$.
We will perform three sets of simulations. In the first two sets of simulations, the power of \ref{gnpa}-\ref{MOM} (GNPMOM) is compared against Bonferroni \citep{bonf36}, HMP and HMPADJ \citep{wilson19}, and F-P \citep{wilson19}. HMP indicates the harmonic mean of $p$-values, whereas HMPADJ indicates the adjusted harmonic mean of $p$-values using the Landau distribution. 

The first set of simulations considers an extremely sparse alternative: $\bm{\mu}=(\mu_1,0,\ldots,0)'$. $\mu_1$ will be selected from the range $[0,3]$, increasing by 0.05 each run. We set the number of iterations $m=10{,}000$; number of test statistics $n=1{,}000$; varying $\rho \in (0.0, 0.2, 0.5, 0.9)$; and varying FWER $\alpha \in (0.01, 0.05, 0.10)$. See results in Figure \ref{Figure 1}. Observe that \ref{gnpa}-\ref{MOM} yields at least equal power compared to other methods. When $\rho$ and $\alpha$ increase we are rewarded with increased power while other methods falter. 

The second set of simulations considers an alternative which begins sparse and becomes more dense. In $\bm{\mu}$, the first $n-s$ values are set to $0$ while the other $n-s+1$ values are set to $\sqrt{\ln(n)}/s^{0.1}$, where $s \in [0,1]$. The same parameters $m,n,\rho,\alpha$ listed the paragraph above were used. For results, refer to Figure \ref{Figure 2}. When the ratio of non-zero means is less than 10\%, \ref{gnpa}-\ref{MOM} appears to attain higher power compared to the other methods. This effect is particularly noticeable for larger $\rho$ and $\alpha$. 

The last simulation considers a $\bm{\mu}$ where $n/2$ test statistics attain mean $3$ while the remaining $n/2$ test statistics attain mean $0$. In this method, we set $m=1$; $n=1{,}000$; $\alpha = 0.10$; and $\rho = 0.5$. The test statistics and plotted and a horizontal red line is imposed to denote the critical value $c_{\alpha}$. See the results in Figure \ref{Figure 3}. Although it is guaranteed to select the maximum $X_{(n)}$ as a significant test statistic, $c_{\alpha}$ struggles to capture other test statistics which exist in the alternative hypothesis. Out of 500 possible test statistics to capture, only 8 are denoted significant. Further refinement of this process is discussed in Section \ref{section:conclusion}.

Lastly, for size analysis of the methods, refer to Appendix \ref{app:table}. In the tables, we see that method \ref{gnpa}-\ref{MOM} provides exact control of $\alpha$ for varying $n$, $\rho$, and $\alpha$, while other methods do not provide exact control everywhere.

\begin{figure}[h!]
    \centering
    \includegraphics[width=1.1\linewidth]{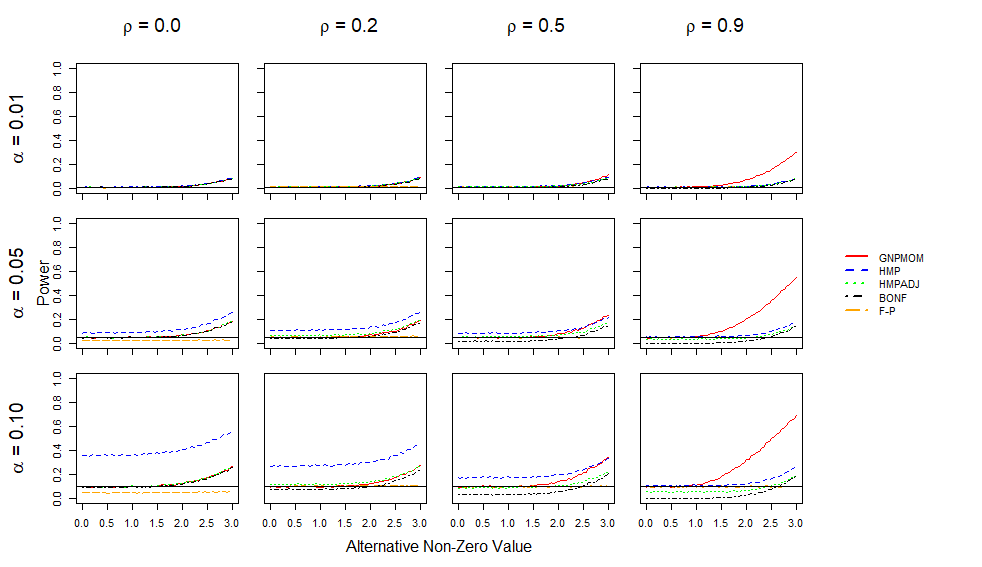}

    \caption{Comparison of the power of \ref{gnpa} against various methods when a single test statistic attains a non-zero mean. Size is indicated by the horizontal black line. Note the higher power achieved by \ref{gnpa} compared to the other methods.}    \label{Figure 1}
\end{figure}
\begin{figure}[h!]
    \centering
    \includegraphics[width=1.1\linewidth]{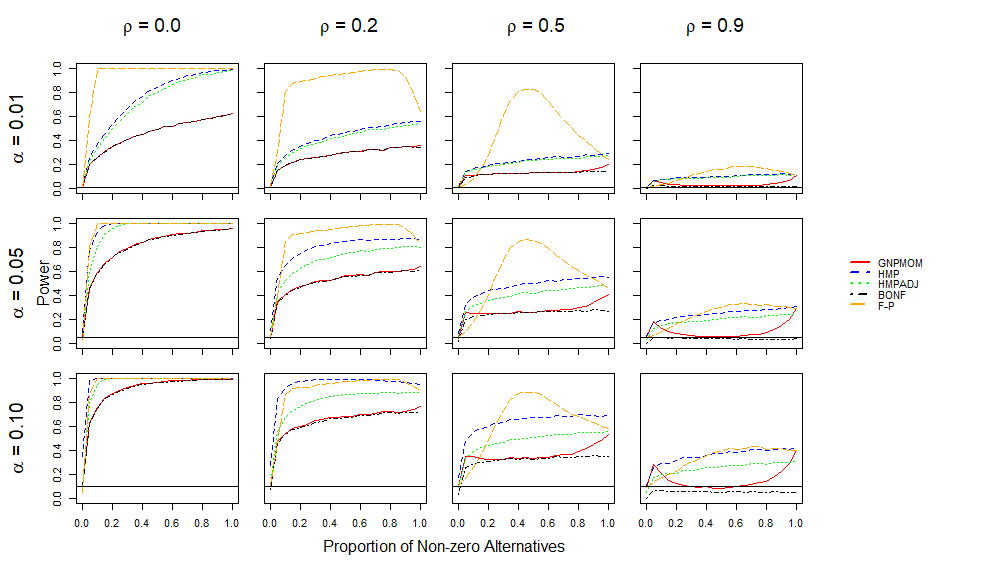}
    
    \caption{Comparison of the power of \ref{gnpa} against various methods as the proportion of test statistics attaining non-zero mean increases. Size is indicated by the horizontal black line. Note the higher power achieved by \ref{gnpa} compared to the other methods when the proportion of non-zero mean test statistics is less than 10\%.}\label{Figure 2}
\end{figure}
\begin{figure}[h]
\centering
\includegraphics[width = 1\linewidth]{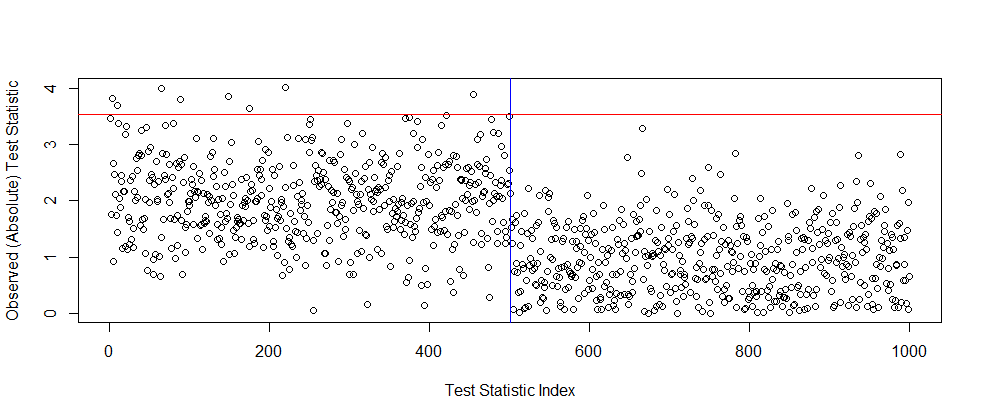}

\caption{Absolute test statistics plotted based on index, where indices less than 500 attain mean 3 while others attain mean 0. The red line indicates the critical value while the blue line indicates the separation between mean 3 and mean 0. Note the conservative nature of the critical value when detecting other test statistics which are not the maximum. }\label{Figure 3}
\end{figure}

\section{Conclusion}\label{section:conclusion}
In this paper, we extended the one-sided findings presented in \cite{gupta73} to the two sided case and provided an implementable form using an estimator described in \cite{follman13}. This method provides a stable control of the FWER at size $\alpha$ as $n \rightarrow \infty$, with reasonable control already at $n=20$. This method is intuitive as a Bonferroni-like procedure which is favorable for practitioners. However, numerical integration may be a barrier for many people. For one-sided alternatives, refer to Table 1 in \cite{gupta73} after estimating $\rho$ with \ref{MOM}. 

The selection process of significant test statistics can be improved. Although guaranteed to yield the maximum test statistic as significant, it fails to detect many others. A couple of sequential-type tests are recommended for further investigation. One method involves directly using the distribution functions for the $r^{th}$ largest variates as described in \cite{gupta73}. The other method removes the maximal data point and reexamining the global hypothesis, repeating until a non-significant test appears. One could also directly find the joint distribution of order statistics $X_{(1)},...,X_{(n)}$ to solve for a critical value $c_{\alpha}$. Further extensions into the $T$ case is desired, particularly the case of $T$ statistics with different realized (independent) $\chi^2$ variates.
\newpage
\newpage
\newpage
\bibliographystyle{chicago}
\bibliography{cites.bib}

\newpage

\appendix
\section{Proofs}\label{app}
\begin{lemma}\label{lemma1}
 Consider a sequence of jointly standard normal random variables $X=(X_1,...,X_n)'$ where
$$E[X_i]=0,\hspace{2cm}E[X_i^2]=1,\hspace{2cm} Corr(X_i,X_j)=\rho$$
for $1\leq i \leq n, 1\leq j \leq n, i\neq j, $ and $0<\rho <1$. 
\newline
\newline
Using the estimator 
$$\widehat{\rho_a} = 1-s_n^2,$$
where $s_n^2$ indicates the (unbiased) sample variance, we have that
$$\widehat{\rho_a} \xrightarrow[n \rightarrow \infty]{P} \rho$$
and
$$\widehat{\rho_a}-\rho = O_{P}\left(\frac{1}{\sqrt{n}} \right).$$
\end{lemma}
\begin{proof}[Proof of Lemma \ref{lemma1}]
 First, we want to show that
$$1-s_n^2 \xrightarrow[n \rightarrow \infty]{P}\rho.$$
The sample variance is defined to be
$$s_n^2=\frac{X' \left( I_n-\frac{1}{n}J_n \right) X}{n-1},$$
where $I_n$ is the identity matrix of size $n$ and $J_n$ is the matrix of $1$s of size $n$. So the expected value is
\begin{equation}
\begin{aligned}
    E[s^2_n]&=\text{tr}\left(\frac{\left(I-\frac{1}{n}J\right)}{n-1}\Sigma\right) \\
    &=1-\rho,
    \end{aligned}
\end{equation}
where $\Sigma$ is the matrix representation of the exchangeable structure and tr$(\cdot)$ is the trace of a matrix. Note then that
$$\frac{(n-1)s_n^2}{1-\rho} \sim \chi^2(n-1)$$
so we have that
\begin{equation}
    \begin{aligned}
        &Var\left( \frac{(n-1)s_n^2}{1-\rho} \right) = 2(n-1) \\
        &\Rightarrow Var(s_n^2) = \frac{2(1-\rho)^2}{(n-1)}.
    \end{aligned}
\end{equation}
So by Chebyshev's, we then have that $\forall \epsilon > 0,$
\begin{equation}
    \begin{aligned}
        P(|1-s_n^2-\rho|\geq \epsilon) &\leq \frac{2(1-\rho)^2}{(n-1)\epsilon^2} \\
        &\xrightarrow[n \rightarrow \infty]{}0,
    \end{aligned}
\end{equation}
which sufficiently shows $\widehat{\rho_a} \xrightarrow[n \rightarrow \infty]{P} \rho$. 
\newline
\newline
To show $\widehat{\rho_a}-\rho = O_P \left (\frac{1}{\sqrt{n}} \right)$, we apply Theorem 14.4-1 in \cite{bishop74}, so
\begin{equation}
    \begin{aligned}
        \widehat{\rho_a}-\rho &=O_P \left( \sqrt{Var(\widehat{\rho_a}}) \right) \\
        &=O_P \left( \frac{1}{\sqrt{n}}\right).
    \end{aligned}
\end{equation}
\end{proof}
\begin{lemma}\label{lemma2}
Consider the sequence of jointly standard normal random variables defined in Lemma 1. If we let $I(\cdot)$ represent the indicator function, we have then
$$I(s_n^2<1) \xrightarrow[n \rightarrow \infty]{P}1$$
which, consequently, provides us that
\begin{enumerate}
    \item $I(s_n^2>1) \xrightarrow[n \rightarrow \infty]{P}0,$
    \item $a_nI(s_n^2>1) \xrightarrow[n \rightarrow \infty]{P} 0\hspace{1cm}\text{where } a_n \in \mathbb{R}^n.$
\end{enumerate}
\end{lemma}
\begin{proof}[Proof of Lemma \ref{lemma2}]
We want to show that
$$I(s_n^2<1) \xrightarrow[n \rightarrow \infty]{P} 1$$
which is equivalent to the statement 
$$P(s_n^2<1) \xrightarrow[n \rightarrow \infty]{}1.$$
Recall that from Lemma \ref{lemma1} that we can adjust $s_n^2$ into a $\chi^2$ variable. Then, we can use the convergence property from Chapter 5 of \cite{nitis00} to show that
\begin{equation}
    \begin{aligned}
        P\left( \frac{(n-1)s_n^2}{1-\rho} < \frac{n-1}{1-\rho} \right) &= P\left( \frac{\frac{(n-1)s_n^2}{1-\rho}-(n-1)}{\sqrt{2(n-1)}} < \frac{\frac{n-1}{1-\rho} -(n-1)}{\sqrt{2(n-1)}} \right) \\
        &\xrightarrow[n \rightarrow \infty]{D}P(Z<\infty) \\
        &=1
    \end{aligned}
\end{equation}
where $Z \sim N(0,1)$. Therefore, we have that $I(s_n^2<1) \xrightarrow[n \rightarrow \infty]{P}1$.
\newline
\newline
Then consequence 1) is straightforward. For consequence 2), we have then
\begin{equation}
    \begin{aligned}
        P \left( \sqrt{a_n} \left|I(s_n^2>1) \right| > \epsilon\right) &= P\left( I(s_n^2>1) > \frac{\epsilon}{\sqrt{a_n}} \right) \\
        &=
        \left\{
        \begin{array}{ll}
0 & \text{ if } I(s_n^2>1) = 0 \\ \\
P(s_n^2>1) & \text{ if } I(s_n^2>1) = 1  \\
\end{array} \right.
    \end{aligned}
\end{equation}
which, from consequence 1), produces convergence to 0 in probability. 
\end{proof}
\begin{lemma}\label{lemma3}
Consider the sequence of jointly standard normal random variables defined in Lemma \ref{lemma1}. Then, we have that
$$\widehat{\rho} = (1-s_n^2)I(s_n^2<1) \xrightarrow[n \rightarrow \infty]{P}\rho$$
\end{lemma}
\begin{proof}[Proof of Lemma \ref{lemma3}]
This is an application of Slutsky's theorem with Lemma \ref{lemma1} and Lemma \ref{lemma2}.
\end{proof}
\begin{proof}[Proof of Theorem \ref{thm1}]
Using Taylor Series on $f(x)=\sqrt{x}$ centered at $1-\rho$, observe that
\begin{equation}
    \begin{aligned}
        &f(1-\widehat{\rho}) = \sqrt{1-\rho} +\frac{1}{2\sqrt{1-\rho}}(\rho-\widehat{\rho})+R_m \\
        &\Rightarrow \sqrt{1-\widehat{\rho}}-\sqrt{1-\rho}=\frac{1}{2\sqrt{1-\rho}}(\rho -\widehat{\rho})+R_m
    \end{aligned}
\end{equation}
where $R_m = \frac{\sum_{m=2}^{\infty}f^{(m)}(1-\rho)(\rho - \widehat{\rho})^{m}}{m!} \sim O_P((p-\widehat{\rho})^m) = O_p\left( \left(\frac{1}{\sqrt{n}}\right)^m \right).$ So to show that Assumption \ref{assumption:rho} holds for $\widehat{\rho}$, we simply have to show that $\sqrt{a_n}(\rho - \widehat{\rho}) \xrightarrow[n \rightarrow \infty]{P}0$. With $\widehat{\rho_a}=1-s_n^2$ as in Lemma \ref{lemma1}, we have that
\begin{equation}
    \begin{aligned}
        \sqrt{a_n}\left(\widehat{\rho} - \rho \right) &= \sqrt{a_n} \left((1-s_n^2)I(s_n^2<1)-\rho -\rho I(s_n^2<1) +\rho I(s_n^2<1  \right) \\
        &=\sqrt{a_n}(\widehat{\rho_a}-\rho)I(s_n^2<1) -\sqrt{a_n}I(s_n^2>1)\rho \\
        &=o\left(\sqrt{n}\right)O_p\left(\frac{1}{\sqrt{{n}}} \right)I(s_n^2<1)-o\left(\sqrt{n}\right)I(s_n^2>1)\rho.
\end{aligned}
\end{equation}
Here, both $o\left( \sqrt{n} \right)O_p\left(\frac{1}{\sqrt{{n}}} \right)I(s_n^2<1)$ and $\sqrt{\ln(n)}I(s_n^2>1)\rho$ converges to 0 in probability through Slutsky's theorem and Lemma \ref{lemma2}. Consequently, $\sqrt{a_n}R_m \xrightarrow[n\rightarrow \infty]{P}0$ as well, so then by the continuous mapping theorem, the entire expression converges to 0 in probability and Assumption \ref{assumption:rho} holds.
\end{proof}
\begin{proof}[Proof of Theorem \ref{thm3}]

First let $F_n(x, \rho, M_n)=\Phi^n\left( \frac{x\sqrt{\rho}+M_n}{\sqrt{1-\rho}} \right)$ for brevity, and see from the mean value theorem that we can find some $c_x$ such that
\begin{equation*}
    \begin{aligned}
        &\left |\int_{-\infty}^{\infty}\left(F(x, \rho, M_n(X)) - F(x, \widehat{\rho}, M_n(X)) \right)\phi(x)dx\right | \\&\leq\int_{-\infty}^{\infty}\left|\left(F(x, \rho, M_n(X)) - F(x, \widehat{\rho}, M_n(X)) \right)\phi(x)\right|dx\\
        &= \int_{-\infty}^{\infty}n(\Phi(c_x))^{n-1}\phi(c_x)\left | \ \frac{x\sqrt{\rho}+M_n(X)}{\sqrt{1-\rho}} - \frac{x\sqrt{\widehat{\rho}}+M_n(X)}{\sqrt{1-\widehat{\rho}}} \right|\phi(x)dx \\
        &\leq \frac{1}{\sqrt{2\pi}} \int_{-\infty}^{\infty}n\Phi^{n-1}(c_x)\left| f_n(x,\hat{\rho},\rho)\right|\phi(x)dx
    \end{aligned}
\end{equation*}
where
\begin{equation*}
    \begin{aligned}
        f_n(x,\widehat{\rho},\rho) = \frac{1}{\sqrt{1-\rho}\sqrt{1-\widehat{\rho}}}\cdot
        (&A_n(\rho)\sqrt{1-\widehat{\rho}}-A_n(\rho)\sqrt{1-\rho}) \\ &-\sqrt{2\ln(n)(1-\rho)}(\sqrt{1-\rho}-\sqrt{1-\widehat{\rho}}) \\ &+x(\sqrt{\rho}\sqrt{1-\widehat{\rho}}-\sqrt{\widehat{\rho}}\sqrt{1-\rho}) 
    \end{aligned}
\end{equation*}
and observe that $n\Phi^{n-1}(c_x) \xrightarrow[n \rightarrow \infty]{} 0$ since $\Phi(c_x) \in (0,1)$. Then, we have
$$f_n(x,\widehat{\rho},\rho) \xrightarrow[n \rightarrow \infty]{D}0$$
pointwise in $x$ through the continuous mapping theorem and the assumption. Because $$F(x, \rho, M_n(X))-F(x, \widehat{\rho}, M_n(X))\xrightarrow[n \rightarrow \infty]{}0$$ pointwise in $x$ and is dominated by $\phi(x)$, we have by the dominated convergence theorem that
$$\lim_{n\rightarrow \infty}\int_{-\infty}^{\infty}n\Phi^{n-1}(c_x)\left| f_n(x,\hat{\rho},\rho)\right|\phi(x)dx=\int_{-\infty}^{\infty}\lim_{n\rightarrow \infty} n\Phi^{n-1}(c_x) \left| f_n(x,\hat{\rho},\rho)\right|\phi(x)dx=0$$
which concludes the proof.
\end{proof}
\begin{lemma}\label{lemma4}
Consider a sequence of jointly standard normal random variables defined in Theorem 2 and define $M_n(X) = \max(X_1,...,X_n)$. Then, we have that
$$P(M_n(|X|)>c_{\alpha}) = 1-\int_{-\infty}^{\infty}\Psi^n_{\mu = \sqrt{\rho}z,\sigma=\sqrt{1-\rho}}(c_\alpha)\phi(z)dz$$
where $\phi(\cdot)$ is the standard normal density function, $\Psi_{\mu, \sigma}(\cdot)$ is the cumulative distribution function of $|Y|$ when $Y \sim N(\mu, \sigma)$, and $c_\alpha$ is some real number.
\end{lemma}
\begin{proof}[Proof of Lemma \ref{lemma4}]
    Since $X_1,...,X_n$ are $N(0,1)$ with equicorrelation $\rho$, we are able to express these $X_i$ as 
$$X_i = \sqrt{\rho}Z+\sqrt{1-\rho}Z_i,$$
where $Z$ and $Z_i$ are $N(0,1)$ and all are jointly independent. Then, we have that
\begin{equation}\label{eqlem4}
    \begin{aligned}
        P(M_n(|X|)>c_\alpha) &=1-P(M_n(|X|)<c_{\alpha}) \\
        &=1-P\left( M_n \left( \left|\sqrt{\rho}Z+\sqrt{1-\rho}Z_i \right| \right) < c_\alpha \right) \\
        &=1-E_{Z}\left[P\left( M_n \left( \left|\sqrt{\rho}Z+\sqrt{1-\rho}Z_i \right| \right) < c_\alpha|Z=z \right)  \right] \\
        &=1-\int_{-\infty}^{\infty}P\left( M_n \left( \left|\sqrt{\rho}z+\sqrt{1-\rho}Z_i \right| \right) < c_\alpha \right) \phi(z)dz.
    \end{aligned}
\end{equation}
Here, since $Z=z$ is given, we have $\sqrt{\rho}z+\sqrt{1-\rho}Z_i \sim N(\sqrt{\rho}z, \sqrt{1-\rho})$ and they are $i.i.d$. We then have
\begin{equation*}
    \begin{aligned}
        (\ref{eqlem4}) &= 1-\int_{\infty}^{\infty} P(M_n(|\sqrt{\rho}z+\sqrt{1-\rho}Z_1|)<c_\alpha)^n\phi(z)dz\\
        &=1-\int_{-\infty}^{\infty}\Psi^n_{\mu = \sqrt{\rho}z,\sigma=\sqrt{1-\rho}}(c_\alpha)\phi(z)dz.
    \end{aligned}
\end{equation*}
\end{proof}
\begin{proof}[Proof of Theorem \ref{thm2}] 
The cumulative distribution function for the folded normal distribution follows 
$$\Psi_{\mu, \sigma}(x) = \frac{1}{2}\left[\text{erf}\left( \frac{x+\mu}{\sigma \sqrt{2}}\right)+\text{erf} \left(\frac{x-\mu}{\sigma \sqrt{2}} \right) \right],$$
where erf$(\cdot)$ is the error function,
$$\text{erf}(x) = \frac{2}{\sqrt{\pi}}\int_{0}^{x}e^{-t^2}dt,$$
so we can show through two consecutive applications of the mean value theorem that
\begin{equation}\label{thm3eq2}
\begin{aligned}
    (\ref{thm3eq1}) &\leq \int_{-\infty}^{\infty}\left|\frac{c_{z,1}^{n-1}n}{2^n}\left(\frac{2e^{-c_{z,2}^2}}{\sqrt{\pi}} g_1(z, \rho, \widehat{\rho})   + \frac{2e^{-c_{z,3}^2}}{\sqrt{\pi}} g_2(z, \rho, \widehat{\rho})\right) \phi(z)\right|dz,
\end{aligned}
\end{equation}
where $c_{z_1}\in (-2,2)$ since $\text{erf}(x) \in (-1,1)$ for all $x$, $c_{z_2}$ and $c_{z_3}$ are in $\mathbb{R}$ dependent on $z$, and
$$g_1(z, \rho, \widehat{\rho}) = \frac{M_n\left( \left| X \right| \right) + \sqrt{\rho}z}{\sqrt{1-\rho}} -\frac{M_n\left( \left| X \right| \right) + \sqrt{\widehat{\rho}}z}{\sqrt{1-\widehat{\rho}}},$$
$$g_2(z, \rho, \widehat{\rho})=\frac{M_n\left( \left| X \right| \right) - \sqrt{\rho}z}{\sqrt{1-\rho}} -\frac{M_n\left( \left| X \right| \right) -\sqrt{\widehat{\rho}}z}{\sqrt{1-\widehat{\rho}}}.$$
Note that since $c_{z_1} \in (-2,2)$, then $\frac{c_{z_1}^{n-1}n}{2^n} \xrightarrow[n \rightarrow \infty]{}0$. Then, for $g_1$, we have that
\begin{equation}\label{thm3eq3}
    \begin{aligned}
        \sqrt{(1-\rho)(1-\widehat{\rho})}g_1(z, \rho, \widehat{\rho}) = M_n(|X|)\left(\sqrt{1-\rho}-\sqrt{1-\widehat{\rho}}\right)+z\left(\sqrt{\rho}\sqrt{1-\widehat{\rho}}-\sqrt{\widehat{\rho}}\sqrt{1-\rho}  \right).
    \end{aligned}
\end{equation}
Note that
\begin{equation*}
    \begin{aligned}
        M_n(|X|) &= \max_i\left(\left|\sqrt{\rho}z+\sqrt{1-\rho}Z_i\right| \right) \\
        &\leq \sqrt{\rho}|z|+\sqrt{1-\rho}\max_{i}(|Z_i|).
    \end{aligned}
\end{equation*}
From \cite{bibinger21}, we have that
$a_n \left(\max_{i}(|Z_i|-b_n \right) \xrightarrow[n\rightarrow \infty]{D}G,$
where
$$a_n = \frac{1}{\sqrt{2\ln(2n)}},$$
$$b_n = \sqrt{2\ln(2n)}-\frac{\ln(4\pi\ln(2n))}{2\sqrt{2\ln(2n)}},$$
and $G$ is the standard Gumbel distribution. Then, 
\begin{equation*}
    \begin{aligned}
        M_n(|X|)\left(\sqrt{1-\rho}-\sqrt{1-\widehat{\rho}} \right)\leq \left( \sqrt{\rho}|z|+\sqrt{1-\rho}\left(a_n\left(\max_{i}|Z_i|-b_n \right)\frac{1}{a_n}+a_nb_n\right)\right)\left(\sqrt{1-\rho}-\sqrt{1-\widehat{\rho}} \right),
    \end{aligned}
\end{equation*}
where then note that
$$\sqrt{1-\rho}-\sqrt{1-\widehat{\rho}} \xrightarrow[n \rightarrow \infty]{P}0,$$
$$\frac{1}{a_n}=o(\sqrt{n}) \Rightarrow \frac{1}{a_n}\left(\sqrt{1-\rho}-\sqrt{1-\widehat{\rho}} \right)\xrightarrow[n\rightarrow \infty]{P}0,$$
$$a_nb_n\xrightarrow[n \rightarrow \infty]{}1,$$
and
$$a_n\left(\max_{i}(|Z_i|)-b_n \right)\left(\sqrt{1-\rho}-\sqrt{1-\widehat{\rho}} \right) \xrightarrow[n \rightarrow \infty]{D}0.$$
Through the continuous mapping theorem and similar arguments for $g_2$, see that
$$(\ref{thm3eq3})\xrightarrow[n \rightarrow \infty]{D}0$$
pointwise in $z$. Then, following the same final steps from Theorem \ref{thm2}, we have our conclusion.
\end{proof}

\section{Size tables} \label{app:table}

\begin{table}[h]
\begin{tabular}{ |c||c|c|c||c|c|c||c|c|c|  }
 \hline
 \multicolumn{10}{|c|}{Empirical Size for $\rho = 0.0$} \\
 \hline
 & \multicolumn{3}{|c||}{$\alpha = 0.01$} &\multicolumn{3}{|c||}{$\alpha = 0.05$}&\multicolumn{3}{|c|}{$\alpha = 0.10$}\\
 \hline
Method & $n$=20 & $n$=100 & $n$=1000 & $n$=20 & $n$=100 & $n$=1000 & $n$=20 & $n$=100 & $n$=1000
\\\hline

 \hline
 GNPMLE   & 0.0103    &0.0099&0.0105&0.0524&0.0492&0.0470&0.1045&0.0967&0.0997\\ \hline
 HMP&   0.0107  & 0.0109   &0.0115&0.0711&0.0751&0.00895&0.1753&0.2230&0.3619\\ \hline
 HMPADJ   &0.0103 & 0.0100&  0.0104&0.0518&0.0481&0.0499&0.1008&0.0951&0.0983\\ \hline
 Bonferroni&   0.0103  & 0.0099&0.0104&0.0506&0.0474&0.0455&0.0989&0.0917&0.0952\\ 
 \hline
 F-P &0.0037&0.0044&0.0042&0.0234&0.0227&0.0233&0.0496&0.0519&0.0474\\
 \hline

\end{tabular}
\newline
\begin{tabular}{ |c||c|c|c||c|c|c||c|c|c|  }
 \hline
 \multicolumn{10}{|c|}{Empirical Size for $\rho = 0.2$} \\
 \hline
 & \multicolumn{3}{|c||}{$\alpha = 0.01$} &\multicolumn{3}{|c||}{$\alpha = 0.05$}&\multicolumn{3}{|c||}{$\alpha = 0.10$}\\
 \hline
Method & $n$=20 & $n$=100 & $n$=1000 & $n$=20 & $n$=100 & $n$=1000 & $n$=20 & $n$=100 & $n$=1000
\\\hline

 \hline
 GNPMLE   &0.0090&0.0081&0.0102&0.0504&0.0504&0.0507&0.1011&0.0965&0.0917\\ \hline
 HMP&0.0106&0.0105&0.0144&0.0716&0.0860&0.1151&0.1647&0.2064&0.2753\\ \hline
 HMPADJ&0.0100&0.0091&0.0127&0.0528&0.0589&0.0689&0.1031&0.1049&0.1159\\ \hline
 Bonferroni&0.0089&0.0078&0.0093&0.0476&0.0460&0.0439&0.0920&0.0831&0.0707\\ 
 \hline
 F-P &0.0639&0.0593&0.0101&0.1241&0.1111&0.0532&0.1681&0.1586&0.1064\\
 \hline
\end{tabular}
\newline
\begin{tabular}{ |c||c|c|c||c|c|c||c|c|c|  }
 \hline
 \multicolumn{10}{|c|}{Empirical Size for $\rho = 0.5$} \\
 \hline
 & \multicolumn{3}{|c||}{$\alpha = 0.01$} &\multicolumn{3}{|c||}{$\alpha = 0.05$}&\multicolumn{3}{|c||}{$\alpha = 0.10$}\\
 \hline
Method & $n$=20 & $n$=100 & $n$=1000 & $n$=20 & $n$=100 & $n$=1000 & $n$=20 & $n$=100 & $n$=1000
\\\hline 

 \hline
 GNPMLE   &0.0091&0.0099&0.0112&0.0460&0.0524&0.0481&0.1004&0.0941&0.1000\\ \hline
 HMP&0.0130&0.0143&0.0172&0.0671&0.0830&0.0844&0.1471&0.1609&0.1764\\ \hline
 HMPADJ&0.0120&0.0133&0.0150&0.0511&0.0607&0.0555&0.0993&0.0931&0.0905\\ \hline
 Bonferroni&0.0080&0.0070&0.0057&0.0347&0.0291&0.0194&0.0690&0.0471&0.0314\\ 
 \hline
 F-P &0.0219&0.0124&0.0102&0.0639&0.0570&0.0493&0.1182&0.0993&0.1010\\
 \hline
\end{tabular}
\newline
\begin{tabular}{ |c||c|c|c||c|c|c||c|c|c|  }
 \hline
 \multicolumn{10}{|c|}{Empirical Size for $\rho = 0.9$} \\
 \hline
 & \multicolumn{3}{|c||}{$\alpha = 0.01$} &\multicolumn{3}{|c||}{$\alpha = 0.05$}&\multicolumn{3}{|c||}{$\alpha = 0.10$}\\
 \hline
Method & $n$=20 & $n$=100 & $n$=1000 & $n$=20 & $n$=100 & $n$=1000 & $n$=20 & $n$=100 & $n$=1000
\\\hline

 \hline
 GNPMLE   &0.0112&0.0095&0.0095&0.0479&0.0502&0.0486&0.0990&0.0977&0.1067\\ \hline
 HMP&0.0119&0.0105&0.0110&0.0523&0.0554&0.0529&0.1070&0.1073&0.1152\\ \hline
 HMPADJ   &0.0112&0.0098&0.0093&0.0388&0.0408&0.0363&0.0745&0.0652&0.0601\\ \hline
 Bonferroni&0.0033&0.0014&0.0005&0.0106&0.0047&0.0008&0.0215&0.0080&0.0022\\ 
 \hline
 F-P &0.0108&0.0096&0.0094&0.0467&0.0523&0.0486&0.1009&0.0974&0.1051\\
 \hline
\end{tabular}
\caption{Empirical size for varying $\rho$, $\alpha$, and $n$. Note that \ref{gnpa} provides exact control of $\alpha$ for each possible combination of parameters, whereas other methods falter in certain combinations. 
}\label{table:sizetable}
\end{table}

\end{document}